\documentclass{article}

\usepackage{arxiv}

\usepackage{microtype}
\usepackage{graphicx}
\usepackage{subfigure}
\usepackage{booktabs} %

\usepackage{hyperref}

\usepackage[utf8]{inputenc} %
\usepackage[T1]{fontenc}    %
\usepackage{hyperref}       %
\usepackage{url}            %
\usepackage{booktabs}       %
\usepackage{amsfonts}       %
\usepackage{nicefrac}       %
\usepackage{microtype}      %
\usepackage{lipsum}		%
\usepackage{graphicx}
\usepackage{natbib}
\usepackage{doi}

\usepackage{amsmath}
\usepackage{amssymb}
\usepackage{mathtools}
\usepackage{amsthm}
\usepackage{xcolor}

\usepackage[capitalize,noabbrev]{cleveref}

\theoremstyle{plain}
\newtheorem{theorem}{Theorem}[section]
\newtheorem{proposition}[theorem]{Proposition}

\theoremstyle{definition}
\newtheorem{definition}[theorem]{Definition}

\theoremstyle{remark}
\newtheorem{remark}[theorem]{Remark}
\newtheorem{example}[theorem]{Example}

\newcommand{\Qset}{\mathbb{Q}}

\newcommand{\DP}{\ensuremath{\mathsf{P}}}
\newcommand{\FP}{\ensuremath{\mathsf{FP}}}
\newcommand{\NP}{\ensuremath{\mathsf{NP}}}

\newcommand{\sharpP}{\ensuremath{\mathsf{\#P}}}

\newcommand{\PM}{\mathcal{PM}}
\newcommand{\cM}{\mathcal{M}}
\newcommand{\RMatch}{\operatorname{RMatch}}
\newcommand{\unsat}{\operatorname{unsat}}
\newcommand{\poly}{\operatorname{poly}}

\newcommand{\sgn}{\operatorname{sgn}}

\newcommand{\bfx}{\mathbf{x}}

\definecolor{bananamania}{rgb}{0.98, 0.91, 0.71}
\definecolor{bleudefrance}{rgb}{0.19, 0.55, 0.91}
\definecolor{rebeccapurple}{rgb}{0.4,0.2,0.6}
\definecolor{fakered}{rgb}{0,0,0}

\usepackage{tikz}
\usetikzlibrary {shapes.geometric} 

\usepackage{enumitem}

\usepackage[textsize=tiny]{todonotes}

\title{Probabilistic Generating Circuits---Demystified$^\ast$}

\author{ Sanyam Agarwal \\
	Saarland University\\
    Saarland Informatics Campus\\
	Saarbr\"ucken, Germany\\
	\texttt{agarwal@cs.uni-saarland.de} \\
	\And
	Markus Bl\"aser \\
	Saarland University\\
    Saarland Informatics Campus\\
	Saarbr\"ucken, Germany\\
	\texttt{mblaeser@cs.uni-saarland.de} \\
}

\date{}

\renewcommand{\headeright}{}
\renewcommand{\undertitle}{}
\renewcommand{\shorttitle}{Probabilistic Generating Circuits---Demystified}

\hypersetup{
pdftitle={Probabilistic Generating Circuits---Demystified},
pdfsubject={cs.CC,cs.AI},
pdfauthor={Sanyam Agarwal, Markus Bl\"aser},
pdfkeywords={Probablistic circuits, probablisitic generating circuits, marginalization, tractable probablistic models},
}

\begin{document}
\maketitle

\begin{abstract}
	Zhang et al.~(ICML 2021, PLMR 139, pp. 12447–1245) introduced probabilistic 
generating circuits (PGCs) as a probabilistic model to unify probabilistic
circuits (PCs) and determinantal point processes (DPPs). At a first glance, PGCs store
a distribution in a very different way, they compute the probability generating polynomial
instead of the probability mass function and it seems that this is the main reason
why PGCs are more powerful than PCs or DPPs. However, PGCs also allow for negative weights,
whereas classical PCs assume that all weights are nonnegative. One of the main insights
of our paper is that the negative weights are responsible for the power of PGCs and not
the different representation. PGCs are PCs in disguise, 
in particular, we show how to transform any PGC into a PC with negative weights
with only polynomial blowup.

PGCs were defined by Zhang et al.\ only for binary
random variables. As our second main result, we show that there is a 
good reason for this: we prove that PGCs for categorial variables with larger
image size do not support tractable marginalization unless $\NP = \DP$.
On the other hand, we show that we can model categorial variables 
with larger image size as PC with negative weights computing set-multilinear
polynomials. These allow for tractable marginalization. 
In this sense, PCs with negative weights strictly subsume PGCs.
\end{abstract}

\renewcommand*{\thefootnote}{\fnsymbol{footnote}}

\footnotetext[1]{\cite{broadrick2024polynomial} independently obtain some of the results
presented in this paper. In particular, they
prove that probabilistic circuits and probabilistic generating 
circuits for binary variables are equivalent 
(our Theorem~\ref{PGCinPCThm}) as well as the hardness of marginalization for probabilistic generating circuits with at least four categories (our Theorem~\ref{NoEMforQuartenary}). These results were obtained independently of ours and were submitted to a conference
around the same time as ours.}

\renewcommand*{\thefootnote}{\arabic{footnote}}

\section{Introduction}
Probabilistic modeling is a central task in machine learning. 
When the underlying models become
large and complicated, however, probabilistic inference easily becomes intractable, see \cite{ROTH1996273} for an explanation.
Therefore, it is important to develop probabilistic models that are tractable (TPMs for short),
that is, they allow for efficient probabilistic inference.
On the other hand, 
the probabilistic models should be as expressive
efficient as possible (in the sense of \citet{DBLP:journals/corr/MartensM14}),
which means that they are able to represent as many different distributions as possible. 
The more classes we can represent,
the broader the spectrum of applications of the model.
There is typically a tradeoff between expressiveness and
tractability. The more expressive the model, the harder will
be probabilistic inference.

Examples of tractable models are for instance bounded treewidth graphical models
\cite{DBLP:journals/jmlr/MeilaJ00, koller}, the well-known determininantal point processes
\cite{Borodin:DPP, DPPmain}, or probabilistic circuits like for instance sum-product networks \cite{Darwiche:Modeling,DBLP:conf/kr/KisaBCD14,DBLP:journals/corr/abs-1202-3732}.
These models represent probability distributions by computing
probability mass functions: the input is an assignment to the random
variables and the output is the corresponding probability of the event.
 
\emph{Probabilistic circuits (PCs)}, see also Section~\ref{sec:PC},  
compute (unnormalized) probability distributions for \emph{syntactic}\footnote{By a syntactic property, we mean a property that can be checked by inspecting the structure of the circuit,
but we do not need to take the input-output behaviour of the circuit into account,
in constrast to semantic properties.}
reasons:
the weights on the edges are required to be nonnegative. 
So-called \emph{decomposable} PCs are tractable again for \emph{syntactic} reasons: The scopes of each product gate are disjoint and therefore marginalization and multiplication commute. 

Another popular tractable model are \emph{determiniantal point processes (DPPs)}. \Citet{DBLP:conf/uai/ZhangHB20} show that PCs in general do not subsume DPPs. 
Shortly after, 
\citet{PGCzhang} propose a new model called 
\emph{probabilistic generating circuits (PGCs)} that represent a distribution
by a generating polynomial. PGCs allow for efficient marginalization and subsume
both PCs and DPPs. It seems that the power of PGCs comes from the fact that they
use a different representation. However, there is another subtle difference
between PGCs and PCs, namely, PGCs allow for negative constants.
\Citet[page 12447]{PGCzhang} write: 
``Because of the presence of negative parameters, it is not guaranteed that the polynomials represented by a PGC is a probability generating polynomial: it might contain terms that are not multiaffine or have negative coefficients''.
Therefore it is a \emph{semantic} property that a PGC represents a probability distribution, the designer of the PGC has to ensure that it computes a 
probability generating function. 

One of the key insights of this paper is that the important property
of PGCs is that they allow for negative constants
and not the fact that they compute a probability generating function instead of 
a probability distribution itself. More precisely, we show how to turn any PGC into a PC 
with negative weights that computes the probability distribution represented by the PGC
and this PC computes a
\emph{set-multilinear} polynomial (see also Definition~\ref{def:setmulti}). The syntactic property of nonnegative weights is replaced by the semantic property of computing a probability distribution and the syntactic property of being decomposable (also called ``syntactically set-multilinear'') 
is replaced by the semantic property of computing a set-multilinear polynomial (but intermediate results might not be set-multilinear).
So PGCs are nothing but PCs in disguise.

\section{Probabilistic circuits}
\label{sec:PC}

An \emph{arithmetic} circuit is a acyclic directed graph. Nodes of indegree $0$ are called
\emph{input nodes}. Internal nodes are either addition nodes or multiplication nodes.
Addition nodes have edge weights on the incoming edges and compute the weighted sum of the inputs.
Product nodes compute the unweighted product of its inputs. 

\emph{Probabilistic circuits (PCs)} are representations of probability distributions that allow tractable inference. See \cite{ProbCirc20} for an in-depth description. They are arithmetic circuits
such that the input nodes correspond to probability
distributions on random variables $X_i$. In the case of binary random variables, we can for instance assume that each distribution
is given as $p x_i + (1-p) \bar x_i$ and therefore the input nodes are labelled by variables
$x_i$ and $\bar x_i$. 
PCs compute probability mass functions.
Figure~\ref{fig:pc:example} shows an example.

\begin{example}
If we set $x_1 = 1$ and $\bar x_1 = 0$ 
as well as $x_2 = 0$ and $\bar x_2 = 1$ in the PC in 
Figure~\ref{fig:pc:example}, we get $\frac 16$ as a result,
which is $\Pr[X_1 = 1, X_2 = 0]$.
\end{example}

The edge weights of a PC are typically assumed to be nonnegative, therefore, PCs
compute (unnormalized) probability distributions by design. We will later also consider PCs
with potentially negative weights. We will call these \emph{nonmonotone} PCs. 
Here, the designer of the PC has to ensure that it computes a probability distribution.

PCs allow for tractable marginalization provided that the PC is \emph{decomposable}
and \emph{smooth}. Decomposable means that for each multiplication gate, the 
scope of the children (i.e., variables it depends on) are disjoint.
A PC is smooth, if for each addition gate the scopes of the children are the same.
While decomposability is a crucial property, smoothness is typically easy
to ensure, see \cite{DBLP:conf/nips/ShihBBA19}.

\begin{example}
The PC in Figure~\ref{fig:pc:example} is decomposable and smooth.
If we set $x_1 = 1$ and $\bar x_1 = 0$ 
as well as $x_2 = \bar x_2 = 1$, we get $\Pr[X_1 = 1] = \frac 12$.
\end{example}

\section{Probablistic generating circuits}

Given categorial variables $X_1,\dots,X_n$ with image $\{0,\dots,d-1\}$
and joint distribution $p(a_1,\dots,a_n) = \Pr[X_1 = a_1, \dots, X_n = a_n]$,
the \emph{probability generating function} is a formal polynomial in formal
variables $z_1,\dots,z_n$ defined by
\begin{align}
   G(z) & = \sum_{j_1 = 0}^{d-1} \dots \sum_{j_n = 0}^{d-1} p(j_1,\dots,j_n) z_1^{j_1} \cdots z_n^{j_n}.
\end{align}
A \emph{probabilistic generating circuit (PGC)} for a probability distribution $p$
is an arithmetic circuit that computes $G$.
While PCs compute probability mass functions, PGCs store probability distributions
as formal objects. We can feed some particular elementary event as an input into a PC and the corresponding
output is the probability of the this event. In PGCs, the probabilities are stored as coefficients.
Figure~\ref{fig:pgc:example} shows an example. %
While PCs can also model continuous distributions, PGCs can only model distributions
over categorial variables due to the way they store distributions.

\begin{figure}[t]
        \centering
        \begin{minipage}[c]{0.3\columnwidth}
        \centering
        \begin{tabular}{l|c|c}
        $\Pr$     & $X_1 = 0$ & $X_1 = 1$ \\
        \hline
        $X_2 = 0$ & $\frac 16$ &  $\frac 16$ \\
        $X_2 = 1$ & $\frac 13$ & $\frac 13$
        \end{tabular} 
        \caption{A distribution over binary random variables}
        \label{fig:distribution}
        \end{minipage}~~~~~~~\begin{minipage}[t]{0.3\columnwidth}
        \centering
       \begin{tikzpicture}[xscale = 0.35, yscale = 0.70]
            \node[shape=circle,draw=black] (1) at (0,0) {$+$};
            \node[shape=circle,draw=black] (2) at (-2,-2) {$\times$};
            \node[shape=circle,draw=black] (3) at (2,-2) {$\times$};
            \node[shape=circle,draw=black] (4) at (-4,-4) {$+$};
            \node[shape=circle,draw=black] (5) at (0,-4) {$+$};
            \node[shape=circle,draw=black] (6) at (4,-4) {$+$};
            \node[shape=circle,draw=black, minimum size = 5ex] (7) at (-5,-6) {$x_1$};
            \node[shape=circle,draw=black, minimum size = 5ex] (8) at (-2.5,-6) {$\Bar{x_1}$};
            \node[shape=circle,draw=black, minimum size = 5ex] (9) at (0.5,-6) {$x_2$};
            \node[shape=circle,draw=black, minimum size = 5ex] (10) at (3.5,-6) {$\Bar{x_2}$};
            \draw[<-] (1) edge node[above left, yshift=-5pt] {$\frac{1}{12}$} (2)
                 (1) edge node[above right, yshift=-5pt] {$ \frac{1}{6}$ } (3);
            \draw[<-] (2) edge (4) (3) edge (4) (2) edge (5) (3) edge (6);
            \draw[<-] (4) edge (7) (4) edge (8) (5) edge (10) (6) edge (9) (5) edge node[left] {$2$} (9) (6) edge node[right] {$\frac{1}{2}$} (10);           
        \end{tikzpicture}
        \caption{A PC over binary random variables, computing the distribution in Figure~\ref{fig:distribution}
        }
        \label{fig:pc:example}
    \end{minipage}~~~~~~~\begin{minipage}[t]{0.3\columnwidth}
        \centering
          \begin{tikzpicture}[xscale = 0.4, yscale = 0.7]
            \node[shape=circle,draw=black] (1) at (0,0) {$+$};
            \node[shape=circle,draw=black] (2) at (-2,-2) {$\times$};
            \node[shape=circle,draw=black] (3) at (2,-2) {$\times$};
            \node[shape=circle,draw=black] (4) at (-4,-4) {$+$};
            \node[shape=circle,draw=black] (5) at (0,-4) {$+$};
            \node[shape=circle,draw=black] (6) at (4,-4) {$+$};
            \node[shape=circle,draw=black, minimum size = 5ex] (7) at (-3,-6) {$1$};
            \node[shape=circle,draw=black, minimum size = 5ex] (8) at (0,-6) {$z_1$};
            \node[shape=circle,draw=black, minimum size = 5ex] (9) at (3,-6) {$z_2$};
            \draw[<-] (1) edge node[above left, yshift = -5pt] {$\frac{2}{3}$} (2)
                 (1) edge node[above right, yshift = -5pt] {$-1$} (3);
            \draw[<-] (2) edge (4) (2) edge (5) (3) edge (5) (3) edge (6);
            \draw[<-] (5) edge (8) (5) edge (7) (4) edge (7) (6) edge (9) (4) edge node[above, xshift = -18pt, yshift = 10pt] {$2$} (9) (6) edge node[above right, xshift = 10pt, yshift = 8pt] {$\frac{1}{2}$} (7);
        \end{tikzpicture}
        \caption{Example of a PGC computing the probability generating function of the distribution in Figure~\ref{fig:distribution}}
        \label{fig:pgc:example}
    \end{minipage}
    \end{figure}

\Citet{PGCzhang} introduce probabilistic generating circuits only for binary random variables
and showed how to perform tractable marginalization in this case. 
While binary variables are the most important categorial random variable,
categorial variables with a larger number of outputs are important, too.
Applications of tenary variables for example 
are excess losses
or labelling with abstention \citep{DBLP:conf/nips/0003S22, DBLP:conf/nips/MhammediGG19, 
DBLP:conf/icml/ThulasidasanBBC19}.

The work by \citet{PGCzhang} led to further work on PGCs: \citet{10.5555/3625834.3625912} 
designed a faster marginalization procedure for PGCs, while \citet{DBLP:conf/icml/Blaser23}
constructed a strongly Rayleigh distribution that cannot be represented by small PGCs.

\section{Our results}

\Citet{PGCzhang} show that given a PGC over binary random variables, marginalization is tractable.
This raises the natural question whether this is also possible for categorial variables attaining more
than two values. We give a negative answer to this question (under standard complexity theoretic assumptions) in Section~\ref{sec:quarternaryhard}.
If marginalization for PGCs over quarternary random variables can be done in polynomial time,
then $\NP = \DP$. For ternary variables, we get a similar result, however, we need to
be able to marginalize over subsets of the image. For the result on quarternary variables,
it is sufficient to marginalize over the whole image set.

\citet{PGCzhang} show that \emph{determinantal point processes} (DPPs) can be represented
by PCGs. On the other hand, it is not clear whether decomposable and smooth PCs (which support tractable marginalization) can represent DPPs. This question is related to proving set-multilinear lower bounds for the determinant in algebraic complexity, see \cite{ramprasad} for an overview.  
See also \cite{DBLP:conf/uai/ZhangHB20} for further limitations of representing DPPs by PCs. 
As our second main result, we here prove that the additional power of the PGC by \citet{PGCzhang} does not come from the fact that they use a different representation but from the fact that one allows for negative constants. 
In particular we prove that every PGC over binary random variables can be transformed into a 
nonmonotone PC, which computes the corresponding
probability mass function and allows for tractable marginalization
(Section~\ref{sec:simulation}).

Third, we prove in Section~\ref{sec:tractable}
that nonmonotone PCs computing set-multilinear polynomials are more general
than PGCs in the sense that they support tractable marginalization over
categorial variables of an arbitrary image size. 
Since computing a set-multilinear polynomial and computing a probability distribution are 
semantic properties (i.e. they cannot be directly inferred from the structure of the PC),
we also ask the question whether checking these properties is hard?
It turns out that the first one can be checked in randomized polynomial time,
while the second property is hard to decide.

In Section~\ref{sec:composition}, we design basic compositional operations,
which preserve the property that a nonmonotone PC computes a probability 
distribution. The first two are the well-known weighted sum
and multiplication (when the domains are disjoint). The third one
is a variant of the hierarchical composition introduced by
\citet{PGCzhang} for PGCs.

Finally, in Section~\ref{sec:DPP}, we discuss the relation
between nonmonotone PCs and DPPs.
PGCs and nonmonotone PCs were designed to subsume monotone PCs and DPPs. 
It is well known that nonmonotone PCs are strictly stronger then monotone ones.
DPPs can only compute distributions with negative correlations.
In this sense, nonmonotone PCs are strictly stronger than DPPs. 
However, once we allow to combine DPPs with 
simple compositional operations like affine projections,
we show that the question whether nonomonotone PCs are more powerful than DPPs
will be very hard to answer (Theorem~\ref{VFinDPP}). It will imply a separation between
algebraic formulas and circuits, a question that has been open for 
decades, see \cite{burgisser_act}.
While it is well-known that we can write a polynomial computed by a formula as 
a determinant, see e.g.~\cite{burgisser_act}, the crucial point here is that the variables
of a DPP only appear on the diagonal. We here prove that every
polynomial computed by a formula can be written as an affine projection
of a DPP of size linear in the formula size.
That means separating DPPs and nonmonotone PCs implies a separation of algebraic 
formulas and circuits. This result can even be strengthened
to so-called algebraic branching programs instead of formulas
(Theorem~\ref{VBPinDPP}).

\section{Graphs and matchings}

In the next two sections, we briefly review definitions and results from graph theory
and computational complexity, which will be needed for our hardness results.

A graph $G$ is called \emph{bipartite}, if we can partition its nodes into two set $U$ and $V$
such that all edges have one node in $U$ and the other node in $V$. 
When we write $G = (U \cup V,E)$ we mean that $G$ is a bipartite graph with bipartition $(U, V)$. 
We will typically call these nodes $u_1,\dots,u_m$ and $v_1,\dots,v_n$.
The \emph{degree} of a node is the number of edges that it is incident to.
A graph is called \emph{regular} if every node has the same degree. 
It is \emph{$d$-regular} if it is regular and the degree of every node is $d$.
The \emph{neighbours} of a node $u$ are the nodes that share an edge with $u$.

\begin{example} Figure~\ref{fig:bip:1} shows a $3$-regular bipartite graph with four nodes
on each side. The neighbours of $u_1$ are $v_1,v_2,v_3$ but not $v_4$.
\end{example}

$M \subseteq E$ is a \emph{matching} if each node of $U$ and $V$ appears in at most 
one edge of $M$. $M$ is called \emph{perfect}, if every node is in exactly one edge. 
If a bipartite graph has a perfect matching, then necessarily $|U| = |V|$.
The size $|M|$ of a matching is the number of edges in it. If $M$ is perfect,
then $|M| = |U| = |V|$. The set of all matchings of $G$ is denoted
by $\cM(G)$ and the set of all perfect matchings by $\PM(G)$.
The number of all matchings and perfect matchings is denoted
by $\#\cM(G)$ and $\#\PM(G)$, respectively.

\begin{example}\
Figure~\ref{fig:bip:2} shows a perfect matching in the graph from 
Figure~\ref{fig:bip:1}.
\end{example}

\begin{figure}
\centering
\begin{minipage}[t]{0.45\columnwidth}
\centering
\begin{tikzpicture}[xscale=1,yscale=0.8]
\tikzstyle{every node}=[circle, draw=black, thick, fill=black, inner sep=2pt];
\node (u1) at (0,0) [circle, draw, fill=bleudefrance, label=left:{$u_1$}] {};
\node (u2) at (0,-1) [circle, draw, fill=bleudefrance, label=left:{$u_2$}] {};
\node (u3) at (0,-2) [circle, draw, fill=bleudefrance, label=left:{$u_3$}] {};
\node (u4) at (0,-3) [circle, draw, fill=bleudefrance, label=left:{$u_4$}] {};
\node (v1) at (2,0) [circle, draw, fill=bananamania, label=right:{$v_1$}] {};
\node (v2) at (2,-1) [circle, draw, fill=bananamania, label=right:{$v_2$}] {};
\node (v3) at (2,-2) [circle, draw, fill=bananamania, label=right:{$v_3$}] {};
\node (v4) at (2,-3) [circle, draw, fill=bananamania, label=right:{$v_4$}] {};
\draw (u1) edge (v1) (u1) edge (v2) (u1) edge (v3)
      (u2) edge (v1) (u2) edge (v2) (u2) edge (v4)
      (u3) edge (v1) (u3) edge (v3) (u3) edge (v4)
      (u4) edge (v2) (u4) edge (v3) (u4) edge (v4);
\end{tikzpicture}
\caption{A $3$-regular bipartite graph with bipartition $U = \{u_1,u_2,u_3,u_4\}$ and 
$V = \{v_1,v_2,v_3,v_4\}$. \label{fig:bip:1}}
\end{minipage}~~~~\begin{minipage}[t]{0.45\columnwidth}
\centering
\begin{tikzpicture}[xscale=1,yscale=0.8]
\tikzstyle{every node}=[circle, draw=black, thick, fill=black, inner sep=2pt];
\node (u1) at (0,0) [circle, draw, fill=bleudefrance, label=left:{$u_1$}] {};
\node (u2) at (0,-1) [circle, draw, fill=bleudefrance, label=left:{$u_2$}] {};
\node (u3) at (0,-2) [circle, draw, fill=bleudefrance, label=left:{$u_3$}] {};
\node (u4) at (0,-3) [circle, draw, fill=bleudefrance, label=left:{$u_4$}] {};
\node (v1) at (2,0) [circle, draw, fill=bananamania, label=right:{$v_1$}] {};
\node (v2) at (2,-1) [circle, draw, fill=bananamania, label=right:{$v_2$}] {};
\node (v3) at (2,-2) [circle, draw, fill=bananamania, label=right:{$v_3$}] {};
\node (v4) at (2,-3) [circle, draw, fill=bananamania, label=right:{$v_4$}] {};
\draw (u1) edge[very thick] (v1) (u1) edge[draw = gray] (v2) (u1) edge[draw = gray] (v3)
      (u2) edge[draw = gray] (v1) (u2) edge[draw = gray] (v2) (u2) edge[very thick] (v4)
      (u3) edge[draw = gray] (v1) (u3) edge[very thick] (v3) (u3) edge[draw = gray] (v4)
      (u4) edge[very thick] (v2) (u4) edge[draw = gray] (v3) (u4) edge[draw = gray] (v4);
\end{tikzpicture}
\caption{The thick edges form a perfect matching. Any subset of it forms a matching.\label{fig:bip:2}}
\end{minipage}
\end{figure}

\section{Complexity theory basics}

We give some background information on the complexity classes 
and results used in this paper.
We refer to \cite{papadimitriou, arora2009computational} for further explanations
and proofs of the well-known definitions and theorems in this section.
The important fact for this paper is that counting perfect matchings in
bipartite graphs is hard.

Deciding whether a formula $\phi$ in conjunctive normal form (CNF) has
a satisfying assignment is the defining problem of the famous class $\NP$.
If we instead want to count the number of satisfying assignments,
we get a problem which is complete for the class $\sharpP$ defined by Valiant.
Obviously, when you can count the number of satisfying assignments,
then you can decide whether there is at least one, therefore, 
$\sharpP$ is a ``harder'' class than $\NP$. It turns out that some problems
become $\sharpP$-hard when considered as a counting problem
while their decision versions are easy. Perfect matchings
in bipartite graphs is such an example: There are efficient algorithms
for the decision problem, but the counting version is hard.

\begin{theorem}[Valiant] Counting perfect matchings in bipartite graphs
is $\sharpP$-complete under Turing reductions.
\end{theorem}

Above, a problem $A$ is Turing reducible to a problem $B$ if there is a polynomial
time deterministic Turing machine that solves $A$ having oracle access to $B$.
This means, that an efficient algorithm for $B$ would yield an efficient algorithm for
$A$ and in this sense, $A$ is easier than $B$. 

Note that $\sharpP$ is a class of functions, not of languages. The 
class $\FP$ is the class of all functions computable in polynomial time.
It relates to $\sharpP$ like $\DP$ relates to $\NP$.
If $\sharpP = \FP$, then $\NP = \DP$.

\section{PGCs do not support efficient marginalization beyond binary variables}

\label{sec:quarternaryhard}

\begin{theorem}\label{NoEMforQuartenary}
    Efficient marginalization over PGCs involving quarternary random variables implies that 
    $\sharpP = \FP$, and in particular, $\NP = \DP$.
\end{theorem}

\begin{proof}
    Counting perfect matchings in $3$-regular bipartite graphs is a $\sharpP$-hard
    problem, as proved by \citet{DAGUM1992283}. 
    Given a $3$-regular bipartite graph $G=(U\cup V,E)$ with $|U| = |V| = n$, we will
    construct a PGC $C$ of size polynomial in $n$ (in fact linear in $n$) 
    over quarternary random variables such that a certain marginal probability 
    is the number of perfect matchings in $G$.

    For every vertex $v_i \in V$, we define a formal variable $V_i$. Similarly, for every edge $e_{i,j}$ representing an edge between $u_i \in U$ and $v_j \in V$, we define a binary random variable $E_{i,j}$. For every vertex $u_i \in U$, let $N(i)$ denote the set of indices of its neighbours. Since $G$ is $3$-regular, $|N(i)| = 3$.
    In particular, let $N(i,1), N(i,2), N(i,3)$ denote the indices of the three neighbours of $u_i$. We now define the polynomial $f$ as 
    \begin{equation} \label{eq:quarternary}
      f(V_1,...,V_n,E_{1,N(1,1)},..., E_{n, N(n,3)}) = \prod_{i=1}^n \sum_{j \in N(i)} E_{i,j} V_j
    \end{equation}
    The right hand side of the equation is a depth-$2$ circuit of size $s = O(n)$, 
    as each sum has only three terms due to $3$-regularity. Thus $f$ has a linear sized circuit.
    
    Furthermore, $f$ is a polynomial of $\deg(f)= 2n$. Each $E_{i,j}$ appears with degree $1$ 
    in $f$, since $E_{i,j}$ only appears in the $i$th factor of the outer product in (\ref{eq:quarternary}). Each $V_j$ appears with degree $3$ in $f$, since each $V_j$ appears in three factors in (\ref{eq:quarternary}) due to $3$-regularity of $G$. The total number of variables appearing in $f$ is $n + 3n = 4n$. 

    Since each coefficient of the product in the inner sums is $1$, the coefficients
    of $f$ are all nonnegative. Thus $f$ can viewed as the probability generating function of the (unnormalized) joint probability distribution of $4n$ quaternary random variables.
    The normalization constant is the sum of all coefficients of $f$, which is just
    $f(1,\dots,1)$. Hence, the normalization constant can be efficiently computed using the circuit for $f$. In this particular case, it is even easier. By looking at the righthand side of (\ref{eq:quarternary}), we see that the normalization constant is $3^n$.
    Thus, $\hat f := f / 3^n$ is a (normalized) probability distribution over quaternary variables and is computed by a linear sized PGC. 

    We now view $\hat f$ as a polynomial in $V_1,\dots,V_n$ with coefficients being polynomials
    in $E_{i,N(i,j)}$, $i = 1,\dots,n$, $j = 1,2,3$. 
    Let $h(E_{1,N(1,1)},..., E_{n, N(n,3)})$ be the coefficient of $V_1 V_2 \dots V_n$.
    We claim that the number of monomials in $h$ gives us the number of perfect matchings in $G$. Any perfect matching in $G$ would be of the form $(u_1,v_{k_1})$, $(u_2,v_{k_2})$, ..., $(u_n,v_{k_n})$, such that $k_i \neq k_j$ for $i \neq j$, and $\{k_1,...,k_n\} = [n]$. But then the monomial $E_{1,k_1}\cdot...\cdot E_{n,k_n}$ would be present in $h$. For the converse, let $E_{1,a_1}\cdot...\cdot E_{n,a_n}$ be some monomial in $h$. Clearly, $a_i\neq a_j$ for $i\neq j$, as each $V_i$ only
    occurs in one factor in (\ref{eq:quarternary}). Further, the coefficient of $E_{1,a_1}\cdot...\cdot E_{n,a_n}$ in $h$ is nonzero only if all the edges $(u_1,v_{a_1})$, $(u_2,v_{a_2}),\dots,(u_n,v_{a_n})$ were present in $G$. Thus, any such monomial would denote a perfect matching in $G$. Hence, the number of monomials in $h$ and the number of perfect matchings in $G$ are equal. 
    
    Suppose we can efficiently marginalize in PGCs over quarternary random variables. 
    Then we can compute $\Pr[V_1 = 1, V_2 = 1,\dots,V_n = 1]$ efficiently.
    (By abuse of notation $V_i$ also denotes the random variable corresponding to the formal variable $V_i$.) However, this probability is nothing but $h(1,\dots,1)$. 
    Each monomial in $h$ has the coefficient $1/3^n$. Thus 
    \[
      \Pr[V_1 = 1,\dots,V_n = 1] = h(1,\dots,1) = \frac{\#\PM(G)}{3^n}.
    \]
     Using marginalization we can efficiently get the number of perfect matchings in a $3$-regular bipartite graph $G$. However, as computing the number of perfect matching in $3$-regular bipartite graphs is a $\sharpP$-complete problem, we get $\sharpP = \FP$.
     This implies $\NP = \DP$.
\end{proof}

\Citet{PGCzhang} show that we can marginalize efficiently over PGCs on binary random variables. Combined with the above result, it naturally raises the question whether there exists efficient marginalization for PGC involving ternary random variables. While we are not able to answer the question fully here, we do show that ``selective'' marginalization over such distributions should not be possible under standard complexity theoretic assumptions. In order to prove this result, we will first introduce some definitions.

\begin{definition}[Selective Marginalization]\label{selmarg}
Let $X_1,\dots,X_n$ be $k$-nary random variables. Let 
$V_i \subseteq \{0,1,\dots,k-1\}$, $1 \le i \le n$. 
$\Pr[X_1 \in V_1,\dots,X_n \in V_n]$ is called a selective marginal probability.
\end{definition}

Let $f(z_1,\dots,z_n)$ be the corresponding probability generating polynomial:
    \[
    f(z_1,...,z_n) = \sum_{s = (s_1,...,s_n) \in \{ 0,1,...,k-1 \}^n}c_s \cdot \prod_{i=1}^n z_i^{s_i}.
    \]
Then we want to compute    
    $\sum_{s_1 \in V_{1},..., s_n \in V_{n}} c_{s}$.

\begin{theorem}\label{NoSelEMforTernary}
    Efficient selective marginalization over PGCs involving ternary random variables implies $\sharpP = \FP$, and in particular, $\NP = \DP$.
\end{theorem}

The proof of the baove theorem will use the following definitions. 
\begin{definition}[$(2,3)$-regular bipartite graph]\label{23regbp}
    Let $G=(L\cup R,E)$ be a bipartite graph such that $|L| = \frac 32 n$ and $|R| = n$. $G$ is a $(2,3)$-regular bipartite graph, if $\forall u \in L$, $\deg(u) = 2$ and $\forall v \in R$, $\deg(v)=3$.
\end{definition}

$n$ has to be even in a $(2,3)$-regular graph and the left-hand size is necessarily $\frac 32 n$ because of the degree constraints. 

\begin{definition}\label{rmatchpoly}
    Let $G = (L\cup R,E)$ be a $(2,3)$-regular bipartite graph. Let $\cM(G)$ be the set of matchings in $G$. Let $M$ be any matching in $\cM(G)$ and $\unsat_R(M)$ denote the set of unmatched $R$-vertices in $M$. Define $\RMatch(G,x)$ as follows:
      \[ 
        \RMatch(G,x) = \sum_{M \in \mathcal{M}(G)} x^{|\unsat_R(M)|}
      \]
\end{definition}

For each fixed $\lambda \in \Qset$, the graph polynomial $\RMatch(G,x)$ defines a mapping $G \mapsto \RMatch(G,\lambda)$, which takes graphs to rational numbers.
\Citet{23reg} show that
the evaluation of the polynomial $\RMatch(G,\lambda)$ is known to be $\sharpP$-hard 
for all $\lambda \in \mathbb{Q}\setminus \{-3,-2,-1,0\}$.

\begin{proof}  (of Theorem~\ref{NoSelEMforTernary}) 
     Consider any $(2,3)$-regular bipartite graph $G=(L \cup R,E)$ with $|L| = m = \frac32 n, |R| = n$. For every vertex $u_i \in L$, we define a ternary random variable $U_i$.  For every vertex $v_j \in R$, let $N(j)$ denote the set of indices of its neighbours. In particular, $N(j,1), N(j,2), N(j,3)$ denote the indices of the three neighbours of $v_j$. We now define the polynomial $f$ as 
     \begin{equation} \label{eq:ternary}
        f(U_1,...,U_m,\lambda) = \prod_{j=1}^n ( \lambda + \sum_{i \in N(j)} U_i) 
     \end{equation}
     where $\lambda$ is an arbitrary positive number, in particular 
     it is not in $\{-3,-2,-1,0\}$. Further, notice that every vertex $u_i \in L$ has degree $ = 2$, hence the degree of any $U_i$ in any monomial of $f$ cannot exceed $2$. We also point out that the coefficient of any monomial of $f$ is 
     nonnegative.
     
     The right-hand side of (\ref{eq:ternary}) is a depth-$2$ circuit of size $s = O(n + m) = O(n)$. Furthermore, $f$ is a polynomial of $\deg(f)=n$. Using proof ideas similar to the proof of Theorem~\ref{NoEMforQuartenary}, we transform $f$ into a probability generating polynomial represented by a PGC by computing $f(1,1,...,\lambda) = (\lambda + 3)^n$. Thus, we set 
     \[
       \hat f(U_1,...,U_m,\lambda) = \frac{f(U_1,...,U_m,\lambda)}{(\lambda + 3)^n}
     \] 
     which has a circuit of size $O(n)$ and represents a probability generating polynomial.
     
     We now selective marginalise $\hat f$ over $V_{i} = \{0,1\}$ for each $i$, as per Definition~\ref{selmarg}. This means that we sum the coefficients of all
     monomials of $\hat f$ that are multilinear. We now claim that the multilinear
     monomials of $\hat f$ stand in one-to-one correspondance with matchings in $G$: Any matching in $G$ would match all $v_j$ for $j \in S \subseteq [n]$ for some $S$, and leave all $v_j$ for $ j \in [n] \setminus S$ unmatched. Let the corresponding matching be $(u_{t_1},v_{s_1})$, ..., $(u_{t_k},v_{s_k})$. Clearly, $u_{t_i} \neq u_{t_j}$ for $i \neq j$. Thus, in the polynomial $f$, this would correspond to the monomial $\lambda^{n-|S|}\cdot \prod_{i \in S}U_{t_i}$. Since all the exponents of $U_i$ are in $\{0,1\}$ for all $i \in [m]$, we see that this term would also be present. Now consider any set of edges $A \subseteq E$ which do not correspond to a matching in $G$. This implies that either some $u_i$ for $i \in [m]$, or some $v_j$ for $j \in [n]$ has been matched more than once. By definition of $f$, for any $j \in [n]$, every monomial of $f$ either takes one neighbour of $v_j$ or leaves $v_j$ unmatched. Further, since we are looking at selective marginalisation with $U_i$ being either $0$ or $1$, this implies none of the $u_i$'s can be matched twice either. Thus there cannot be any monomial that corresponds to $A$. Therefore, the result of the selective marginalization is 
     \[
        \frac{\sum_{S \in \cM(G)} \lambda^{n-|S|}} {(\lambda + 3)^n} = 
           \frac{\RMatch(G,\lambda)} {(\lambda + 3)^n} .
     \]
     Thus, using \emph{selective} marginalisation we can efficiently determine the $\RMatch$ polynomial for a given $\lambda$ in a $(2,3)$-regular bipartite graph $G$. 
     However, computing this is a $\sharpP$-complete problem. Recall that the PGC for $\hat f$ has size $O(m)$. Thus, efficient selective marginalization over ternary PGCs would imply that $\sharpP = \FP$.
\end{proof}

\section{PGCs over binary variables can be simulated by nonmonotone PCs}

\label{sec:simulation}

\begin{theorem}{\label{PGCinPCThm}}
    PGCs over binary variables can be simulated by nonmonotone PCs with only polynomial overhead
    in size.
\end{theorem}

\begin{proof}
    Let $f(z_1,...,z_n)$ be a probability generating function of $n$ binary random variables
    computed by a PGC of size $s$.
    We define the polynomial $g$ as 
    \[ 
       g(x_1,\Bar{x_1},...,x_n,\Bar{x_n}) = f(\frac{x_1}{\Bar{x_1}},\frac{x_2}{\Bar{x_2}},...,\frac{x_n}{\Bar{x_n}})\cdot \prod_{i=1}^n \Bar{x_i}.
    \]
    $g$ is indeed a polynomial, since $f$ is multilinear.
    To get a circuit for $g$, first we need to substitute all instances of $z_i$ with $\frac{x_i}{\Bar{x_i}}$ in the circuit for $f$, using the help of division gates, which we will remove later. Clearly, this gives us a circuit for $f(\frac{x_1}{\Bar{x_1}},\frac{x_2}{\Bar{x_2}},...,\frac{x_n}{\Bar{x_n}})$ with size $O(s+n)$. Furthermore, we can compute $\prod_{i=1}^n \Bar{x_i}$ in size $O(n)$, and thus we get a circuit for $g$ of size $O(s+n)$.
    
    Let $m$ be any monomial of $f(z_1,z_2,...,z_n)$. Then, $m = c_S \cdot \prod_{i\in S} z_i$
    for some $S \subseteq [n]$. Looking at the corresponding monomial $m'$ in $g(x_1,\Bar{x_1},...,x_n,\Bar{x_n})$, we see that \[m' = c_S \cdot (\prod_{i \in S} \frac{x_i}{\Bar{x_i}})\cdot (\prod_{j=1}^n \Bar{x_j}) = c_S (\prod_{i \in S} x_i) \cdot (\prod_{i \notin S} \Bar{x_i})\]
    Thus, if $z_i$ occurs in $m$, then $x_i$ occurs in $m'$. Similarly, if $z_i$ does not occur in $m$, then $\Bar{x_i}$ occurs in $m'$. Thus, in $m'$ exactly one of $x_i$ or $\Bar{x_i}$ occurs for all $i$. Hence, $m'$ is multilinear with $\deg(m') = n$. Since, $m'$ was an arbitrary monomial of $g$, this implies $g$ is both multilinear and homogenous with $\deg(g)=n$. Thus, we obtained a 
    nonmonotone PC, cf. \citep[Proposition 1]{PGCzhang}.
    
    Right now the circuit of $g$ has division gates. 
    This is a problem, since we might want to set $\bar x_i = 0$ when computing a probability.
    However using the famous result of \citet{StrassenDiv}, we can eliminate all division gates to get a circuit for $g$ of size $\poly(s+n, \deg(g)) = \poly(s,n)$. Furthermore, $g$ computes a probability distribution since $f$ represented a probability distribution. 
    Thus, starting with a PGC of size $s$, we get a nonmonotone PC computing a set-multilinear $g$ 
    of size $\poly(s,n)$.
\end{proof}

\begin{remark}
In the proof by Strassen, every arithmetic operation in the original circuit is replaced by
a corresponding operation on polynomials of degree $n$. Therefore, 
the exact upper bound on the size of the new circuit is $O(s \cdot n \log n \log \log n)$
if we use the fast polynomial multiplication by \citet{DBLP:journals/acta/CantorK91}
or $O(s \cdot n)$ if we use interpolation.
\end{remark}

\begin{example}
The following small example explains Strassen's construction: 
Suppose we have a PGC that computes the following polynomial: $f(z_1,z_2) = 0.6z_1z_2 + 0.4z_1$.
We will demonstrate the conversion of the polynomial into a polynomial computed by the corresponding nonmonotone PC. 
\begin{itemize}[noitemsep,nolistsep]
        \item We first replace each $z_i$ with $\frac{x_i}{\Bar{x_i}}$, and then get $g = f(\frac{x_1}{\Bar{x_1}}, \frac{x_2}{\Bar{x_2}}) \cdot \Bar{x_1}\Bar{x_2} = 0.6x_1x_2 + 0.4x_1\Bar{x_2}$. %
        However, the new circuit contains
        the two divisions $\frac{x_1}{\Bar{x_1}}$ and $\frac{x_2}{\Bar{x_2}}$, which we have to remove.        
        \item The idea by \citet{StrassenDiv} is to expand $\frac 1{\bar x_1}$ as a formal power series. Since we are computing a polynomial in the end, it is enough to work with finite approximations. We can only invert a power series if it has a nonzero constant term. Therefore, we first have to perform a Taylor shift, which will be reverted in the end.
        In our case, the shift $\Bar{x_i}$ to $1-\Bar{x_i}$ works. We get
        $f(\frac{x_1}{1 - \Bar{x_1}}, \frac{x_2}{1 - \Bar{x_2}}) \cdot (1 - \Bar{x_1})(1 - \Bar{x_2}) = 0.6x_1x_2 + 0.4x_1 (1 - \Bar{x_2})$.
        \item The inverse of $1 - x$ as a formal power series is the geometric series
        $\frac 1{1-x} = \sum_{i = 0}^{\infty} x^i$. Since in this example, we compute a multilinear function of degree two, it is enough to work with order-one approximations and 
        replace $\frac{1}{1 - \Bar{x_i}}$ by $1 + \Bar{x_i}$. We get
        \begin{align}
        \lefteqn{f(x_1 (1 + \Bar{x_1}), x_2 (1 + \Bar{x_2})) \cdot (1 - \Bar{x_1})(1 - \Bar{x_2})} & \notag \\
          & =  (0.6 x_1 x_2 ( 1 + \Bar{x_1}) (1+ \Bar{x_2}) + 0.4 x_1 ( 1+ \Bar{x_1}) ) 
             (1-\Bar{x_1})(1 - \Bar{x_2}) \notag \\
          & = 0.6 x_1 x_2 + 0.4 x_1 ( 1 - \bar x_2) ~+~ \text{higher degree terms}
                 \label{eq:strassen:1}.
        \end{align}
        \item To compute only the terms up to degree two and remove the unwanted higher degree terms, we compute the homogeneous parts separately. This is the part where we incur the blowup in the circuit size. 
        \item Now, finally to get back to $g$, we need to invert the substitutions that we made earlier. Hence, replacing $\Bar{x_i}$ with $1-\Bar{x_i}$ for $i \in \{1,2\}$ in (\ref{eq:strassen:1}) after removing the higher degree terms gives us the polynomial $0.6x_1x_2 +0.4x_1\Bar{x_2}$ which is in fact $g$.
    \end{itemize}
\end{example}

\section{Nonmonotone PCs computing set-multilinear polynomials 
support tractable marginalization}

\label{sec:tractable}

We consider categorial random variables $X_1,\dots,X_n$, w.l.o.g.\
taking values in $I = \{0,1,\dots,d-1\}$. With each random variable $X_i$,
we associate $d$ indeterminates $z_{i,0},\dots,z_{i,d-1}$. A probability distribution
for $X_i$ given by $\Pr[ X_i = \delta] = \alpha_{\delta}$, $0 \le \delta < d$,
can be modelled by the linear polynomial 
$\ell = \sum_{\delta = 0}^{d-1} \alpha_\delta z_{i,\delta}$. By setting 
$z_{i,\delta}$ to $1$ and all other $z_{i,\delta'}$ for $\delta' \not= \delta$
to $0$, that is, evaluating $p_i$ at the unit vector $e_\delta$,
we get $\ell (e_\delta) = \alpha_{\delta}$. More generally,
if we want to compute $\Pr[X_i \in S]$ for some set
$S \subseteq \{0,\dots,d-1\}$, we can compute this by evaluating $\ell$
at $v$, where $v_i = 1$ if $i \in S$ and $v_i = 0$ otherwise.

\begin{definition} \label{def:setmulti}
Let $X$ be a set of variables and $Y_1,\dots,Y_t$ be a partition of $X$.
A polynomial $p$ in variables $X$ is called \emph{set-multilinear} with respect
to the above partition, if every monomial of $p$ contains exactly one
variable from each set $Y_\tau$, $1 \le \tau \le t$, with degree $1$.
In particular, $p$ is homogeneous of degree $t$.
\end{definition}

\begin{example}
A decomposable and smooth PC over binary variables computes a set-multilinear polynomial with the parts given by $\{x_i, \bar x_i\}$,
$1 \le i \le n$.
\end{example}

Given a decomposable and smooth PC over categorial random variables $X_1,\dots,X_n$, we can model every input distribution by a linear form as described above. The corresponding polynomial
will be set-multilinear with the parts of the partition being
$\{z_{i,0},\dots,z_{i,d-1}\}$, $1 \le i \le n$. However, when the PC is nonmonotone,
that is, we allow for negative weights, then we can compute a set-multilinear polynomial,
even if the PC is not decomposable. It turns out that this is sufficient for
performing marginalization.

\begin{theorem}
Let $C$ be a nonmonotone PC of size $s$ 
computing a probability distribution
over categorial random variables $X_1,\dots,X_n$
such that the polynomial $P$ computed by $C$ is set-multilinear with respect
to the partition $\{z_{i,0},\dots,z_{i,d-1}\}$, $1 \le i \le n$.
Let $A_1,\dots,A_n \subseteq \{0,\dots,d-1\}$. Then we can compute
$\Pr[X_1 \in A_1,\dots,X_n \in A_n]$ in time $O(s)$.
\end{theorem}

\begin{proof}
Note that
\begin{equation}
  P = \sum_{j_1 = 0}^{d-1} \dots \sum_{j_n = 0}^{d-1} \Pr[X_1 = j_1, \dots, X_n = j_n] z_{1,j_1} \cdots z_{n,j_n}.  \label{lastequation}
\end{equation}
Consider the elementary event $X_1 = a_1,\dots,X_n = a_n$. Define an input vector $e$ for $P$ by
\[
  e_{i,j} = \begin{cases}
               1 & \text{if $j = a_i$}, \\
               0 & \text{otherwise}
            \end{cases}
\]
for $1 \le i \le n$, $0 \le j \le d-1$. By (\ref{lastequation}),
$P(e) = \Pr[X_1 = a_1,\dots,X_n = a_n]$. From all monomials in $P$, only 
$z_{1,a_1} \dots z_{n,a_n}$ evaluates to $1$ under $e$
and all others evaluate to $0$. Now to compute $\Pr[X_1 \in A_1,\dots,X_n \in A_n]$,
we simply evaluate at the point 
\[
  v_{i,j} = \begin{cases}
               1 & \text{if $j \in A_i$}, \\
               0 & \text{otherwise}.
            \end{cases}
\]
We claim that $P(v) = \Pr[X_1 \in A_1,\dots,X_n \in A_n]$: A monomial  
$z_{1,j_1} \cdots z_{n,j_n}$ evaluated at $v$ becomes $1$ iff $j_i \in A_i$ for all 
$1 \le i \le n$. Otherwise, it evaluates to $0$. Thus
\begin{align*}
  P(v) & = \sum_{j_1 \in A_1} \dots \sum_{j_n \in A_n} \Pr[X_1 = j_1, \dots, X_n = j_n] \cdot 1\\
       & = \Pr[X_1 \in A_1,\dots,X_n \in A_n],  
\end{align*}
which proves the claim.
\end{proof}

\begin{remark}
This gives also an alternative marginalization algorithm for PGCs over $n$ binary variables.
We convert it into a equivalent nonmonotone PC computing a set-multilinear polynomial.
The total running time will be $O(s\cdot n)$. The extra factor $n$ comes from the conversion
to a PC. The running time matches the one by \citet{10.5555/3625834.3625912}.
\end{remark}

While being decomposable or being smooth is a \emph{syntactic} property, that is, it is a property of the circuit and can be checked efficiently, computing a set-multilinear polynomial and computing a probability distribution are \emph{semantic} properties\footnote{This is only true for nonmonotone PCs, for monotone PCs these conditions are equivalent, see \cite{vergari2020probabilistic}.}, 
that is, properties of the polynomial computed by the circuit. We can of course compute the coefficients of the polynomial to check whether it is set-multilinear or evaluate the circuit at all inputs to see whether it is a probability distribution, but this requires exponential time.
Can we check these properties nevertheless efficiently? It turns out that the first property is efficiently checkable while the second is most likely not.

\begin{proposition} \label{prob:testmultilin}
Testing whether a nonmonotone PC computes a set-multilinear polynomial 
with respect to a given partition
can be done in randomized polynomial time.
\end{proposition}

\begin{proof} %
Let $P$ be the polynomial computed by a given circuit $C$ and let
$\{z_{i,0},\dots,z_{i,d-1}\}$, $1 \le i \le n$, be the parts of the partition.

First we check whether each monomial of $P$ depends on at least one variable 
of each part. To this aim we iterate over all $i$ and set
$z_{i,0} = \dots = z_{i,d-1} = 0$. If the resulting polynomial is nonzero,
then there is at least one monomial in $P$ that does not contain a variable
from $\{z_{i,0},\dots,z_{i,d-1}\}$. Testing whether a polynomial
is nonzero can be done in randomized polynomial time
by the Schwartz-Zippel-Lemma, see \cite{arora2009computational}.

Second we need to check that each monomial depends on at most one variable
from each part and has degree $\le 1$ in this variable. 
For each pair $z_{i,j},z_{i,j'}$ with $j \not= j'$, we replace all other
variables by random values (from a polynomially large set) obtaining
a polynomial $P_{i,j,j'}$. 
Then again by the Schwartz-Zippel-Lemma, if the polynomial $P$ has a monomial containing $z_{i,j}$ and $z_{i,j'}$ or $z_{i,j}^2$ or $z_{i,j'}^2$,
then also the new polynomial $P_{i,j,j'}$ has such a monomial with high probability. 
Since $P_{i,j,j'}$ has only two
variables, we can expand it completely and check whether is has a monomial containing $z_{i,j}$ and $z_{i,j'}$ or $z_{i,j}^2$ or $z_{i,j'}^2$.

The overall procedure is polynomial time.
\end{proof}

\begin{proposition}
Testing whether a nonmonotone PC computes a probability distribution is
$\NP$-hard.
\end{proposition}

The proof of the above proposition essentially follows from the hardness result by
\citet[Theorem 5]{10.5555/3625834.3625912} combined with our transformation of 
PGCs to PCs (Theorem~\ref{PGCinPCThm}).

\section{Compositional operations for nonmonotone PCs}

\label{sec:composition}

Since it is hard to determine whether a nonmonotone PC computes a probability distribution
we here present three compositional operations for nonmonotone PCs that preserve
the property that the PC computes a probability distribution and computes
a set-multilinear polynomial. The first two are the well-known multiplication and mixing operations. The third one is more interesting, it is a variant
of the hierarchical composition for PGC introduced by \citet{PGCzhang}.

Let $X_1,\dots,X_n$ be categorial random variables with image
$\Delta = \{0,\dots,d-1\}$. 
Let $A, B \subseteq \{1,\dots,n\}$.
Let $C, D$ be two nonmonotone PCs computing joint probability distributions $f$ and $g$
for $X_A = (X_i)_{i \in A}$ and $X_B = (X_j)_{j \in B}$, resp., as set-multilinear 
polynomials in the variables $z_{i,j}$, $1 \le i \le n$, $0 \le j \le d-1$.
That is
\[
  f(a) = \sum_{a \in \Delta^{|A|}} \alpha_a \prod_{i \in A} z_{i, a_i}, \quad 
  g(b) = \sum_{b \in \Delta^{|B|}} \beta_b \prod_{j \in B} z_{j, b_j},
\]
where $\alpha_a = \Pr[X_A = a]$ and $\beta_b = \Pr[X_b = b]$.
Let $s$ and $t$ be the sizes of $C$ and $D$ respectively.

We want to construct a mixture of the two distributions. To this aim, we can extend $f$ to a probability distribution
on $X_{A \cup B}$ by
\[
   f(a,b')  =  \left(\sum_{a \in \Delta^{|A|}} \alpha_a \prod_{i \in A} z_{i, a_i}  \right) \prod_{j \in B \setminus A} \frac 1d (z_{j,0} + \dots z_{j,d-1}),
\]
where $b'$ corresponds to the variables $X_{B \setminus A}$.
In the same way, we can extend $g$. This is a kind of ``smoothening'' operation,
which ensures that both distributions have the same domain.

\begin{proposition} \label{prop:mix}
There is a nonmonotone PC of size $O(s + t +nd)$ computing the 
mixture $\alpha f + (1-\alpha) g$ on $X_{A \cup B}$  for any $0 \le \alpha \le 1$
as a set-multilinear polynomial. 
\end{proposition}
\begin{proof} %
The extended function $f(a,b')$ is set-multilinear. The same is true for $g$.
A linear combination of set-multilinear polynomials on the same partition is set-multilinear.
The additional $O(nd)$ term comes from the smoothening operation.
\end{proof}

\begin{proposition} \label{prop:mult}
If $A$ and $B$ are disjoint, then there is a nonmonotone PC of size $O(s + t)$
computing the product distribution $f(a) \cdot g(b)$ as a set-multilinear polynomial.
\end{proposition}

\begin{proof} %
Since $A$ and $B$ are disjoint,
\begin{align*}
   f(a) g(b) & =  \left( \sum_{a \in \Delta^{|A|}} \alpha_a \prod_{i \in A} z_{i, a_i} \right)
                  \left( \sum_{b \in \Delta^{|B|}} \beta_b \prod_{j \in B} z_{j, b_j} \right) \\
      & =  \sum_{a \in \Delta^{|A|}, b \in \Delta^{|B|}} \alpha_a \beta_b \prod_{i \in A} z_{i, a_i} \prod_{j \in B} z_{j, b_j}.
\end{align*}
\end{proof}

Finally, we can also define a hierarchical composition. 
Let $f$ be a distribution on binary random variables
given as a set-multilinear polynomials in variables
$z_1,\dots,z_n, \bar z_1,\dots, \bar z_n$.
Let $g_1,\dots,g_n$ be distributions on $m$ disjoint $d$-ary random variables each
given as set-multilinar polynomials in variables
$y_{i,j,k}$, $1 \le i \le n$, $1 \le j \le m$, $0 \le k \le d-1$.

\begin{definition}
The hierarchical composition of $f$ and $g_1,\dots,g_n$ is defined by the
set-multilinear polynomial obtained by replacing 
$z_i$ by $g_i$ and $\bar z_i$ by $\prod_{j = 1}^m \sum_{k = 1}^d \frac{1}{d} y_{i,j,k}$,
which is the uniform distribution on the domain of $g_i$.
\end{definition}

\begin{proposition} \label{prop:hier}
The hierarchical composition is indeed a probability distribution.
It can be computed by a nonmonotone PC whose size is linear in
the sum of the sizes of the nonmonotone PCs for $f$ and $g_1,\dots,g_n$.
\end{proposition}

\begin{proof} (or Proposition~\ref{prop:hier})
The hierarchical composition is a mixture of $2^n$ many $n$-fold products of distributions
on disjoint variables. The proposition follows immediately by repeated application of 
Propositions~\ref{prop:mix} and~\ref{prop:mult}.
\end{proof}

\section{Nonmonotone PCs computing set-multilinear polynomials versus DPPs}

\label{sec:DPP}

\emph{Determinantal point processes (DPPs)} are stochastic point processes whose probability distribution is characterized by a determinant of some matrix. DPPs are important 
because they are able to express negative dependencies. For the purposes of modeling real data, the class of DPPs is restricted to \emph{L-ensembles} \cite{DPPmain}, which have the interesting property:

\begin{remark}
    Let $X_1, X_2,...,X_n$ be binary random variables.
    For any $Y \subseteq [n]$, the marginal probability is given by:
    $\Pr(X_i = 1 \; | \; i \in Y) = \det(L + I_{\Bar{Y}})$
    where $L$ is the \textit{L-ensemble} matrix of the DPP, and $I_{\Bar{Y}}$ denotes the diagonal matrix with all entries indexed by elements of $Y$ as 0, and the rest as 1.
\end{remark}

As a polynomial, a DPP computes a polynomial $\det(L + \mathbf{X})$, where
$\mathbf{X}$ is a diagonal matrix with variables on the diagonal.

PGCs were designed by \citet{PGCzhang} to subsume decomposable PCs and DPPs.
It is natural to ask whether PGCs strictly subsume DPPs. The obvious
answer is ``yes'', since DPPs can only model negative dependencies. 
However, what happens if we allow simple pre- or postprocessing?
We prove that this question will be hard to answer:

\begin{theorem}[Formulas as DPPs]{\label{VFinDPP}}
    Any arithmetic formula can be represented as an affine  projection of a \emph{DPP}.
\end{theorem}

An arithmetic formula is an arithmetic circuit whose underlying structure is a tree.
An affine projection is a mapping that maps the variables to affine linear forms
in (a subset of) the variables.

The interpretation of the Theorem~\ref{VFinDPP} is the following. Assume there is a PGC that we cannot write as a projection of a DPP. Then, since every arithmetic formula is a projection of a DPP, we found an arithmetic circuit (the PGC) that cannot be written as an arithmetic formula.
This problem is open in algebraic complexity theory for decades, see \cite{burgisser_act}.
It is well-known that every formula is the projection
of a determinant, see \cite{burgisser_act}. However, this is not sufficient to answer
our question, since all known constructions place the variables in off-diagonal entries.

\begin{remark}
In a DPP, the matrix $L$ should be positive semi-definite. Theorem~\ref{VFinDPP} 
does not ensure this. However, we can add a large value $m$ to the diagonal elements
of the matrix constructed in Theorem~\ref{VFinDPP} and make is diagonally dominant
(and hence positive definite). The value $m$ can be subtracted again by the affine projection.
\end{remark}

Before we start with the proof, we need some combinatorial interpretation
of the determinant:
The determinant of an $n \times n$-matrix $A = (a_{i,j})$ is defined by %
\[
  \det A = \sum_{\pi \in S_n} \sgn(\pi) a_{1,\pi(1)} \dots a_{n,\pi(n)}.
\]
One can think of $A$ being the weighted adjacency matrix of directed graph.
A permutation $\pi$ then corresponds to a cycle cover in the graph: 
A \emph{cycle cover} is a collection of node-disjoint directed cycles such that
each node appears in exactly one cycle. This is nothing but the cycle decomposition
of the permutation. The sign can be written as $\sgn(\pi) = (-1)^{n + \#\text{cycles}}$.
The weight $w(C)$ of a cycle cover $C$ is the product of the weights of the edges in it multiplied
with the sign. In this way, we can write the determinant as the sum of the weights
of all cycle covers. If $G$ is the directed graph corresponding to the matrix $A$,
we denote this sum by $w(G)$. Thus $\det A = w(G)$.

\newcommand{\Cov}{\operatorname{Cov}}

\begin{proof} (of Theorem~\ref{VFinDPP})
    The proof will be by induction on the structure of the formula.
    For every subformula, we will create a corresponding directed graph $G=(V,E)$ with self-loops and a unique start vertex $s$ and end vertex $t$. Each edge of the graph will have a nonnegative weight assigned to it. An $s,t$-cover of such a graph is a set of edges consisting of a directed path from $s$ to $t$ and directed cycles such that each nodes is either in the path or in exactly one of the cycles. The weight of such an $s,t$-cover
    is the weight of the cycle cover that we get when we add the back edge $(t,s)$.
    
    For a formula computing a polynomial
    $p(x_1,\dots,x_n)$, we construct a graph $G$ such that it has the following properties:
    \begin{enumerate}[noitemsep,nolistsep]
        \item $w(G \setminus \{s,t\}) = 1$ and there is exactly one cover. 
        \item $w(G \setminus \{s\}) = w(G \setminus \{t\}) = 0$ and in both cases 
        there are no covers.
        \item $\displaystyle \sum_{C \in \Cov(G)} w(C) = p(x_1,...,x_n)$ where $\Cov(G)$ are all valid $s,t$-coverings of $G$.
     \end{enumerate}
    Furthermore the variables $x_1,\dots,x_n$ only appear on self-loops. 
    The proof is by structural induction. The base case is when the formula is a constant or a variable and in the induction step, we combine two smaller formulas using an addition or multiplication gate.
    
    \textbf{Base Case}: $p(\bfx) = v$ where $v$ can be either a field constant or a single variable.
    The corresponding graph $G$ is shown in Figure~\ref{vf2dpp_basis}.
    Note that, since $v$ could be a variable, it occurs only on a self loop. 
    
    \begin{figure}[t]
        \centering     
        \begin{minipage}{0.45\columnwidth}
        \centering
        \begin{tikzpicture}[xscale = 0.6, yscale = 0.52]
            \node[circle, draw=black, thick, fill=bleudefrance, inner sep=2pt, label=left:{$s$}] (s) at (0,0) {};
            \node[circle, draw=black, thick, fill=bananamania, inner sep=2pt, label=left:{$2$}] (2) at (0,-2) {};
            \node[circle, draw=black, thick, fill=bleudefrance, inner sep=2pt, label=left:{$t$}] (t) at (0,-4) {};
            \node[circle, draw=black, thick, fill=bananamania, inner sep=2pt, label=above:{$3$}] (3) at (2,-1) {};
            \node[circle, draw=black, thick, fill=bananamania, inner sep=2pt, label=below:{$4$}] (4) at (2,-3) {};
            \draw (s) edge[->] (2)
                (2) edge[->] (t) (2) edge[->, bend left] (3)
                (3) edge[->, bend left] (4) (3) edge[->, loop right, looseness = 10, in = 30, out = -30] node {$v$} (3)
                (4) edge[->, bend left] (2) (4) edge[->, loop right, looseness = 10, in = -30, out = 30] (4);
        \end{tikzpicture}
       \end{minipage}~~~~ \begin{minipage}{0.45\columnwidth}
       \centering
        $ \displaystyle \begin{bmatrix}
        0 & 1 & 0 & 0 & 0 \\
        0 & 0 & 1 & 0 & 1 \\
        0 & 0 & v & 1 & 0 \\
        0 & 1 & 0 & 1 & 0 \\
        1 & 0 & 0 & 0 & 0 \\
        \end{bmatrix}$
        \end{minipage}
        \caption{The graph in the base case and the corresponding adjacency matrix.
        The source $s$ and target $t$ are drawn in blue, internal nodes are drawn in yellow.
        $v$ only appears on the diagonal. Edges without label have weight $1$.}
        \label{vf2dpp_basis}
    \end{figure}
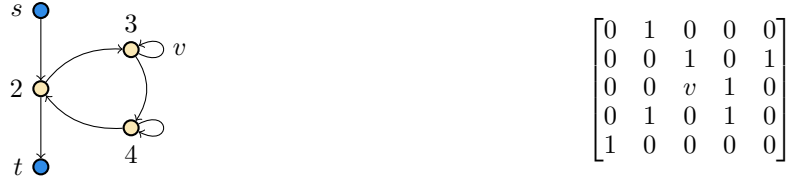
    
    \begin{enumerate}[noitemsep,nolistsep]
        \item $G' = G \setminus \{s,t\}$ leaves us with a graph $G'$ containing only a $3$-cycle of nodes $2, 3$ and $4$. The only way to cover node $2$ is to include it in a cycle with nodes $3$ and $4$, and hence the only possible cycle cover with nonzero weight, in fact weight $1$, is $(234)$. 
        \item Consider $G' = G \setminus \{s\}$. Clearly, $t$ in $G'$ is not part of any cycle. Hence, any cycle cover would need to include $t$ with a self loop. But $t$ 
        has no self loop. Hence, $w(G') = 0$. Similarly, $w(G \setminus \{t\}) = 0$.
        \item Looking at the covers $C$ of $G$ with an $s,t$-path and disjoint cycles, we notice that there is just one possible cover: the path $s,2,t$ and two self loops on vertices $3$ and $4$. The sign is $(-1)^{5 + 3} = 1$ and the weight is $v$.
    \end{enumerate}

\textbf{Inductive case 1}: $p(\bfx) = p_1(\bfx) + p_2(\bfx)$. 
    Let $G_i$ be the graph for the formula $p_i(\bfx)$ with start and end vertices being $s_i$ and $t_i$ respectively. The corresponding graph $G$ will be constructed 
    by identifying $s_1$ with $s_2$ and $t_1$ and $t_2$, see Figure~\ref{vf2dpp_addn}    
    for a schematic drawing.

    \begin{figure}[t]
    \centering
    \begin{minipage}[t]{0.45\columnwidth}
        \centering
        \begin{tikzpicture}
            \node[circle, draw=black, thick, fill=bleudefrance, inner sep=2pt, label=above:{$s=s_1=s_2$}] (s) at (0,0) {};
            \node[circle, draw=black, thick, fill=bleudefrance, inner sep=2pt, label=below:{$t=t_1=t_2$}] (t) at (0,-2) {};
            \node (G1) at (-0.6,-1) {$G_1$};
            \node (G2) at (0.6,-1) {$G_2$};
            \draw (s) edge[bend right=30] (t) (s) edge[bend right=60, looseness=1.7] (t) (s) edge[bend left=30] (t) (s) edge[bend left=60, looseness=1.7] (t);
        \end{tikzpicture}
        \caption{The construction in the case of addition. The source nodes and the target nodes of both graph are identified with each other.}
        \label{vf2dpp_addn}
    \end{minipage}~~~~~\begin{minipage}[t]{0.45\columnwidth}    
        \centering
        \begin{tikzpicture}
            \node[shape=circle,draw=black, thick, fill=bleudefrance, inner sep = 2pt, label=right:{$s=s_1$}] (s) at (0,0) {};
            \node (G1) at (0,-0.75) {$G_1$};
            \node[shape=circle,draw=black, thick, fill=bananamania, inner sep = 2pt, label=right:{$t_1$}] (t1) at (0,-1.5) {};
            \node[shape=circle,draw=black, thick, fill=bananamania, inner sep = 2pt, label=left:{$z$}] (z) at (-0.5,-2) {};
            \node[shape=circle,draw=black, thick, fill=bananamania, inner sep = 2pt, label=right:{$s_2$}] (s2) at (0,-2.5) {};
            \node (G2) at (0,-3.25) {$G_2$};
            \node[shape=circle,draw=black, thick, fill=bleudefrance, inner sep = 2pt, label=right:{$t=t_1$}] (t) at (0,-4) {};
            \draw (s) edge[bend right = 45] (t1) (s) edge[bend left = 45] (t1);
            \draw (t1) edge[->, bend right] (z) (z) edge[->, bend right] (s2) (s2) edge[->] (t1) (s2) edge[bend right = 45] (t) (s2) edge[bend left = 45] (t);
        \end{tikzpicture}
        \caption{The construction in the case of multiplication. The source $s_1$ of the $G_1$
        becomes the new source $s$ and the target $t_2$ of $G_2$ becomes the new target $t$. 
        $t_1$ and $s_2$ become internal nodes.}
        \label{vf2dpp_mult}
    \end{minipage}    
    \end{figure}
    
    \begin{enumerate}[noitemsep,nolistsep]
        \item $G' = G \setminus \{s,t\}$. Since $G_1' = G_1 \setminus \{s,t\}$ and $G_2' = G_2\setminus \{s,t\}$ are not connected, the only possible cycle covers of $G'$ are the disjoint union of cycle covers of $G_1'$ and $G_2'$. By, induction 
        their weights will be equal to $1$ and there is only one cover in each graph.  
        Hence, $w(G') = 1$ and there is only one cover of $G'$.
        \item Consider $G' = G \setminus \{s\}$. Now, there are two possible ways to cover $G'$. Either $t$ is covered in $G_1'$ or in $G_2'$. Suppose, $t$ is covered in $G_1'$. But $G_1'$ has no cycle cover containing $t$. The same is true for when 
        $t$ is covered in $G_2'$. Thus there is no cover of $G'$.
        The case for $G' = G \setminus \{t\}$ is analogous.
        \item Looking at the covers $C$ of $G$ with an $s,t$-path and disjoint cycles, we notice that either $G_1$ or $G_2$ contains the $s,t$-path. In the first case, the $s,t$-cover of
        $G$ will be disjoint unions of $s,t$-covers of $G_1$ and of cycle covers
        of $G_2 \setminus \{s,t\}$. There is only one 
        cycle cover of $G_2 \setminus \{s,t\}$ and it has weight $1$.
        Thus the sum of all such $s,t$-covers is $p_1(\bfx)$. Similarly,
        when the $s,t$-path is contained in $G_2$, we get the weight  $p_2(\bfx)$.
        Hence, $w(G) = p_1(\bfx)\cdot 1 + 1 \cdot p_2(\bfx) = p(\bfx)$.
    \end{enumerate}
    
    \textbf{Inductive case 2}: $p(\bfx) = p_1(\bfx) \cdot p_2(\bfx)$. Let $G_i$ be the graph for the formula $p_i(\bfx)$ with start and end vertices being $s_i$ and $t_i$ respectively. The corresponding graph $G$ will be constructed as shown in
    Figure~\ref{vf2dpp_mult}: The start vertex is $s = s_1$ and the end vertex is $t = t_2$.
    $t_1$ and $s_2$ are connected via a $3$-cycle.
   
    \begin{enumerate}[noitemsep,nolistsep]
        \item $G' = G \setminus \{s,t\}$. Now, any cycle cover of $G'$ will have $z$ covered in the $(z,s_2,t_1)$-cycle. This means that $t_1$ will not be a part of $G_1$ and $s_2$ will not be a part of $G_2$. Hence $G_1' = G_1 \setminus \{s_1,t_1\}$ and $G_2' = G_2\setminus \{s_2,t_2\}$. These two graph have only one cover each and their
        weight is $1$. Since we add one cycle $(z,s_2,t_1)$ and odd number of nodes,
        the overall weight of the whole cycle cover is $1$, again.

        \item $G' = G \setminus \{s\}$. Like the previous case, $z$ can only be covered
        in the $(z,s_2,t_1)$-cycle. This means that $t_1$ will not be a part of $G_1$ and $s_2$ will not be a part of $G_2$. Let $G_1' = G_1 \setminus \{s_1,t_1\}$ and $G_2' = G_2\setminus \{s_2\}$. By the induction hypothesis, $G_2'$ has no cycle covers
        and thus $G'$ has neither.
        The case for $G' = G \setminus \{t\}$ is analogous.
        
        \item Looking at the covers $C$ of $G$ with an $s,t$-path and disjoint cycles, we notice that there we can only go from $s$ to $t$ using $z$.         
        Any such cover $C$ will include a path $P = s=s_1\rightsquigarrow t_1 \rightarrow z \rightarrow s_2 \rightsquigarrow t=t_2$. Let $P_i  = s_i\rightsquigarrow t_i$,
        $i = 1,2$.        
        Any cycle cover formed by joining $t$ to $s$ with a weight $1$ edge will include the path $P$ as a cycle and further covers $C_1$ and $C_2$ of disjoint odd cycles from $G_1\setminus P$ and $G_2 \setminus P$. This cover has one cycle less than the two corresponding covers in $G_1$ and $G_2$,
        since the two path $P_1$ and $P_2$ in $G_1$ and $G_2$ are merged into $P$. On the other
        hand, we also have one more node, $z$. Therefore, the sign of $C$ is the product of 
        the signs of $P_1 \cup C_1$ and $P_2 \cup C_2$
        and the weight $w(C)$ 
        is $w(P_1 \cup C_1) \cdot w(P_2 \cup C_2)$. 
        Since we can take any combination of $s,t$-covers of $G_1$ and $G_2$,
        $w(G) = p_1(\bfx) \cdot p_2(\bfx) = p(\bfx)$.
    \end{enumerate}
   
Now given a formula, we get an equivalent DPP by applying our inductive construction %
and adding the %
edge $(t,s)$.
\end{proof}

We now generalize Theorem~\ref{VFinDPP} to a more powerful model, the so-called algebraic branching programs. An ABP is a acyclic graph with a source $s$ and a sink $t$. Edges are labeled with constants or variables. The weight of an $s,t$-path is the product of the weights of the edge in the path. The polynomial computed by an ABP is sum of the weights of all $s,t$-path.
ABPs can be efficiently simulated by arithmetric circuits, since they can be written as an iterated matrix multiplication. See \cite{burgisser_act} for more information. Generalizing the construction of the graph in Theorem~\ref{VFinDPP}, we can even show that any ABP can be represented as a DPP. 

\newcommand{\IMM}{\operatorname{IMM}}

\begin{theorem}[ABPs as DPPs]{\label{VBPinDPP}}
    An ABP of size $s$ can be represented as a DPP of size $\poly(s)$.
\end{theorem}

\begin{proof}
    Recall that any ABP is a projection of an iterated matrix multiplication (IMM). Hence, we will show a reduction from the iterated matrix multiplication polynomial
    $\IMM_{n,d}$, which is the $(1,1)$-entry of $d$ variable
    matrices of size $n \times n$, to a graph whose 
    determinant represents a DPP. Our construction uses ideas similar to the ones used in proof of Theorem~\ref{VFinDPP}. We will modify the ABP in such a way that any $s,t$-path in our new graph will have an odd number of vertices. Further, any other cycle in the graph would either be a self-loop or a $3$-cycle. Hence, on connecting $t$ to $s$ via an edge, any node in the graph will always be covered by an odd cycle. Hence, the sign of any cycle cover is always positive. We show our construction for the ABP  corresponding to $\IMM_{2,3}$ in Figure~\ref{imm2,3}.
    
    \begin{figure}[t]
       \centering
       \begin{minipage}[b]{0.45\textwidth}
        \centering
        \begin{tikzpicture}[scale = 1]
        \definecolor{cadmiumgreen}{rgb}{0.0, 0.42, 0.24}
            \node[shape=circle,draw=black, thick, fill=bleudefrance, inner sep = 2pt, label={[text=bleudefrance]left:{$s$}}] (s) at (0,0) {};
            \node[shape=circle,draw=black, thick, fill=cadmiumgreen!50, inner sep = 2pt, label={[text=cadmiumgreen]above:{$a_{11}$}}] (a1) at (2,2) {};
            \node[shape=circle,draw=black, thick, fill=cadmiumgreen!50, inner sep = 2pt, label={[text=cadmiumgreen]below:{$a_{12}$}}] (a2) at (2,-2) {};
            \node[shape=circle,draw=black, thick, fill=cadmiumgreen!50, inner sep = 2pt, label={[text=cadmiumgreen]above:{$a_{21}$}}] (b1) at (4,2) {};
            \node[shape=circle,draw=black, thick, fill=cadmiumgreen!50, inner sep = 2pt, label={[text=cadmiumgreen]below:{$a_{22}$}}] (b2) at (4,-2) {};
            \node[shape=circle,draw=black, thick, fill=bleudefrance, inner sep = 2pt, label={[text=bleudefrance]right:{$t$}}] (t) at (6,0) {};

            \draw (s) edge[->] node[above left] {$x_{11}^{(1)}$} (a1)
                 (s) edge[->] node[below left] {$x_{12}^{(1)}$} (a2);
            
            \draw (a1) edge[->] node[above] {$x_{11}^{(2)}$} (b1)
                 (a1) edge[->] node[right, pos=0.7] {$x_{12}^{(2)}$} (b2);

            \draw (a2) edge[->] node[left, pos=0.3] {$x_{21}^{(2)}$} (b1)
                 (a2) edge[->] node[below] {$x_{22}^{(2)}$} (b2);

            \draw (b1) edge[->] node[above right] {$x_{11}^{(3)}$} (t)
                 (b2) edge[->] node[below right] {$x_{21}^{(3)}$} (t);

        \end{tikzpicture}
        \caption{ABP computing  $\IMM_{2,3}$}
        \label{imm2,3}
  \end{minipage}~~~~~~\begin{minipage}[b]{0.45\textwidth}
        \centering
        \begin{tikzpicture}[xscale = 0.7]
        \definecolor{cadmiumgreen}{rgb}{0.0, 0.42, 0.24}
            \node[shape=circle,draw=black, thick, fill=bleudefrance, inner sep = 2pt, label={[text=bleudefrance]left:{$s$}}] (s) at (0,0) {};
            \node[shape=circle,draw=black, thick, fill=bananamania!50, inner sep = 2pt] (sa1) at (1.5,1) {};
            \node[shape=circle,draw=black, thick, fill=bananamania!50, inner sep = 2pt] (sa2) at (1.5,-1) {};
            \node[shape=circle,draw=black, thick, fill=cadmiumgreen!50, inner sep = 2pt, label={[text=cadmiumgreen]below right:{$a_{11}$}}] (a1) at (3,2) {};
            \node[shape=circle,draw=black, thick, fill=cadmiumgreen!50, inner sep = 2pt, label={[text=cadmiumgreen, label distance=1mm]above:{$a_{12}$}}] (a2) at (3,-2) {};
            \node[shape=circle,draw=black, thick, fill=bananamania!50, inner sep = 2pt] (a1b1) at (4.5,2) {};
            \node[shape=circle,draw=black, thick, fill=bananamania!50, inner sep = 2pt] (a2b2) at (4.5,-2) {};
            \node[shape=circle,draw=black, thick, fill=cadmiumgreen!50, inner sep = 2pt, label={[text=cadmiumgreen]230:{$a_{21}$}}] (b1) at (6,2) {};
            \node[shape=circle,draw=black, thick, fill=cadmiumgreen!50, inner sep = 2pt, label={[text=cadmiumgreen, label distance=1mm]above:{$a_{22}$}}] (b2) at (6,-2) {};
            \node[shape=circle,draw=black, thick, fill=bananamania!50, inner sep = 2pt] (b1t) at (7.5,1) {};
            \node[shape=circle,draw=black, thick, fill=bananamania!50, inner sep = 2pt] (b2t) at (7.5,-1) {};
            \node[shape=circle,draw=black, thick, fill=bleudefrance, inner sep = 2pt, label={[text=bleudefrance]right:{$t$}}] (t) at (9,0) {};
            \node[shape=circle,draw=black, thick, fill=bananamania!50, inner sep = 2pt] (a1b2) at (3.5,0.5) {};
            \node[shape=circle,draw=black, thick, fill=bananamania!50, inner sep = 2pt] (a2b1) at (5.5,0.5) {};

            \node[shape=circle,draw=black, thick, fill=bananamania!50, inner sep = 2pt] (z1) at (1.1,1.5) {};
            \node[shape=circle,draw=black, thick, fill=bananamania!50, inner sep = 2pt] (z2) at (1.7,1.8) {};

            \node[shape=circle,draw=black, thick, fill=bananamania!50, inner sep = 2pt] (z3) at (1.1,-1.5) {};
            \node[shape=circle,draw=black, thick, fill=bananamania!50, inner sep = 2pt] (z4) at (1.7,-1.8) {};

            \node[shape=circle,draw=black, thick, fill=bananamania!50, inner sep = 2pt] (z5) at (4.2,2.5) {};
            \node[shape=circle,draw=black, thick, fill=bananamania!50, inner sep = 2pt] (z6) at (4.8,2.5) {};

            \node[shape=circle,draw=black, thick, fill=bananamania!50, inner sep = 2pt] (z7) at (4.2,-2.5) {};
            \node[shape=circle,draw=black, thick, fill=bananamania!50, inner sep = 2pt] (z8) at (4.8,-2.5) {};

            \node[shape=circle,draw=black, thick, fill=bananamania!50, inner sep = 2pt] (z9) at (7.3,1.8) {};
            \node[shape=circle,draw=black, thick, fill=bananamania!50, inner sep = 2pt] (z10) at (7.9,1.5) {};

            \node[shape=circle,draw=black, thick, fill=bananamania!50, inner sep = 2pt] (z11) at (7.3,-1.8) {};
            \node[shape=circle,draw=black, thick, fill=bananamania!50, inner sep = 2pt] (z12) at (7.9,-1.5) {};

            \node[shape=circle,draw=black, thick, fill=bananamania!50, inner sep = 2pt] (m1) at (2.8,0.6) {};
            \node[shape=circle,draw=black, thick, fill=bananamania!50, inner sep = 2pt] (m2) at (3,-0.2) {};

            \node[shape=circle,draw=black, thick, fill=bananamania!50, inner sep = 2pt] (m3) at (6.2,0.6) {};
            \node[shape=circle,draw=black, thick, fill=bananamania!50, inner sep = 2pt] (m4) at (6,-0.2) {};

            \node[shape=circle,draw=black, thick, fill=bananamania!50, inner sep = 2pt] (k1) at (2.7,2.5) {};
            \node[shape=circle,draw=black, thick, fill=bananamania!50, inner sep = 2pt] (k2) at (3.3,2.5) {};

            \node[shape=circle,draw=black, thick, fill=bananamania!50, inner sep = 2pt] (k3) at (2.7,-2.5) {};
            \node[shape=circle,draw=black, thick, fill=bananamania!50, inner sep = 2pt] (k4) at (3.3,-2.5) {};
            
            \node[shape=circle,draw=black, thick, fill=bananamania!50, inner sep = 2pt] (k5) at (5.7,2.5) {};
            \node[shape=circle,draw=black, thick, fill=bananamania!50, inner sep = 2pt] (k6) at (6.3,2.5) {};
            
            \node[shape=circle,draw=black, thick, fill=bananamania!50, inner sep = 2pt] (k7) at (5.7,-2.5) {};
            \node[shape=circle,draw=black, thick, fill=bananamania!50, inner sep = 2pt] (k8) at (6.3,-2.5) {};
            \draw (s) edge[->] (sa1) (s) edge[->] (sa2) (sa1) edge[->] (a1) (sa2) edge[->] (a2) (a1) edge[->] (a1b1) (a1) edge[->] (a1b2) (a2) edge[->] (a2b2) (a2) edge[->] (a2b1) (a1b1) edge[->] (b1) (a2b2) edge[->] (b2) (a1b2) edge[->] (b2) (a2b1) edge[->] (b1) (b1) edge[->] (b1t) (b2) edge[->] (b2t) (b1t) edge[->] (t) (b2t) edge[->] (t) (t) edge[->, bend left=105, looseness=1.6] (s);
            \draw (sa1) edge[->] (z1) (z1) edge[->] (z2) (z2) edge[->] (sa1) (z1) edge[->, loop above] node[above] {$x_{11}^{(1)}$} (z1) (z2) edge[->, loop above] (z2);
            \draw (sa2) edge[->] (z3) (z3) edge[->] (z4) (z4) edge[->] (sa2) (z3) edge[->, loop below] node[below] {$x_{12}^{(1)}$} (z3) (z4) edge[->, loop below] (z4);
            \draw (a1b1) edge[->] (z5) (z5) edge[->] (z6) (z6) edge[->] (a1b1) (z5) edge[->, loop above] node[above] {$x_{11}^{(2)}$} (z5) (z6) edge[->, loop above] (z6);
            \draw (a2b2) edge[->] (z7) (z7) edge[->] (z8) (z8) edge[->] (a2b2) (z7) edge[->, loop below] node[below] {$x_{22}^{(2)}$} (z7) (z8) edge[->, loop below] (z8);
            \draw (b1t) edge[->] (z9) (z9) edge[->] (z10) (z10) edge[->] (b1t) (z9) edge[->, loop above] node[above] {$x_{11}^{(3)}$} (z9) (z10) edge[->, loop above] (z10);
            \draw (b2t) edge[->] (z11) (z11) edge[->] (z12) (z12) edge[->] (b2t) (z11) edge[->, loop below] node[below] {$x_{21}^{(3)}$} (z11) (z12) edge[->, loop below] (z12);
            \draw (a1b2) edge[->] (m2) (m2) edge[->] (m1) (m1) edge[->] (a1b2) (m2) edge[->, loop left] node[left] {$x_{12}^{(2)}$} (m2) (m1) edge[->, loop left] (m1);
            \draw (a2b1) edge[->] (m4) (m4) edge[->] (m3) (m3) edge[->] (a2b1) (m4) edge[->, loop right] node[right] {$x_{21}^{(2)}$} (m4) (m3) edge[->, loop right] (m3);
            \draw (a1) edge[->] (k1) (k1) edge[->] (k2) (k2) edge[->] (a1) (k1) edge[->, loop above] (k1) (k2) edge[->, loop above] (k2);
            \draw (a2) edge[->] (k3) (k3) edge[->] (k4) (k4) edge[->] (a2) (k3) edge[->, loop below] (k3) (k4) edge[->, loop below] (k4);
            \draw (b1) edge[->] (k5) (k5) edge[->] (k6) (k6) edge[->] (b1) (k5) edge[->, loop above] (k5) (k6) edge[->, loop above] (k6);
            \draw (b2) edge[->] (k7) (k7) edge[->] (k8) (k8) edge[->] (b2) (k7) edge[->, loop below] (k7) (k8) edge[->, loop below] (k8);            
        \end{tikzpicture}
        \vspace*{-30pt}
        \caption{DPP graph $G$ for $\IMM_{2,3}$}
        \label{imm_dpp}
        \end{minipage}
    \end{figure}
    
    For an ABP corresponding to $\IMM_{n,d}$ we create the graph $G$ with the following properties:
    \begin{itemize}[noitemsep,nolistsep]
        \item We include all the nodes in the ABP in our graph $G$. We create an edge form $t$ to $s$ in $G$.
        \item Let $e=(u,v)$ be an edge in the ABP with weight $w_e$. In $G$, we add a node $n_e$ corresponding to $e$. We join $u$ and $n_e$ with a directed edge of weight $1$ going from $u$ to $n_e$, and join $n_e$ to $v$ with a directed edge of weight $1$. Furthermore, we create two auxiliary nodes $n_e^{(1)}$ and $n_e^{(2)}$ with self-loops of weight $w_e$ and $1$ respectively, and form a $3$-cycle between $n_e, n_e^{(1)}, n_e^{(2)}$ by adding edges $(n_e, n_e^{(1)}), (n_e^{(1)}, n_e^{(2)})$ and $(n_e^{(2)}, n_e)$, each of weight $1$.
        \item Furthermore, for every nodes $u$ in the ABP except for the start and end nodes, we add two auxiliary nodes $u^{(1)}$ and $u^{(2)}$ with self-loops of weight $1$. We finally create a $3$-cycle between $u, u^{(1)}, u^{(2)}$ by adding directed edges $(u, u^{(1)}), (u^{(1)}, u^{(2)}), (u^{(2)}, u)$, each of weight $1$.
    \end{itemize}
  The corresponding DPP graph for the ABP in Figure~\ref{imm2,3} is shown in Figure~\ref{imm_dpp}.
  
    Now, let $P = s\rightarrow a_{1,k_1} \rightarrow a_{2,k_2} \rightarrow ... \rightarrow a_{d-1,k_{d-1}} \rightarrow t$ be any $s\rightsquigarrow t$ path in the original ABP. Then, the corresponding path $P'$ in $G$ will be of the form $s \rightarrow n_{(s,a_{1,k_1})} \rightarrow a_{1,k_1} \rightarrow n_{(a_{1,k_1}, a_{2,k_2})} \rightarrow a_{2,k_2} \rightarrow ... \rightarrow a_{d-1,k_{d-1}} \rightarrow n_{(a_{d-1,k_{d-1}},t)} \rightarrow t$. Hence, it is easy to see that any such path will have an odd number of vertices. Further, the weight of this path is $1$ in $G$. Any cycle cover in $G$ must cover $s$ and $t$, which can only be done by a $s,t$-path in $G$ along with the edge $(t,s)$. Thus, any such cycle will have weight $=1$, and an odd number of vertices implying a positive sign. Furthermore, for all the nodes $u$ of the original ABP not occurring on the path $P'$, they can only be covered by the corresponding $3$-cycles $(u , u^{(1)} , u^{(2)} )$, which have weight $1$ and a positive sign. Moreover, for all the nodes that appear on the path $P'$ and were present in the original ABP, the corresponding auxiliary nodes will have to be covered by self-loops of weight $1$. Finally, for the nodes in $P'$ that do not appear in the $ABP$ (like $n_{(a_{i,k_i}, a_{i+1,k_{i+1}})}$), the corresponding auxiliary nodes will have to be covered by self-loops of weight  $w_{(a_{i,k_i}, a_{i+1,k_{i+1}})}$ and $1$ respectively. Hence, the sign of any such cycle cover will be positive, as it only has $3$-cycles and self-loops, and the weight would just be 
    $w_{(s, a_{1,k_1})} \cdot w_{(a_{1,k_1}, a_{2,k_{2}})} \cdot ... \cdot w_{(a_{d-1,k_{d-1}}, t)}  
       = x_{1k_1}^{(1)} \cdot x_{k_1k_2}^{(2)} \cdot ... \cdot x_{k_{d-1}1}^{(d)}$.
    Thus, 
    \begin{align*}
      \det(G) & = \sum_{C} \sgn(C) w(C) \\
              & = \sum_{C} 1 \cdot x_{1k_1}^{(1)} \cdot x_{k_1k_2}^{(2)} \dots x_{k_{d-1}1}^{(d)} \\
              & =  \sum_{\text{$s,t$-path}} x_{1k_1}^{(1)} \cdot x_{k_1k_2}^{(2)} \cdot ... \cdot x_{k_{d-1}1}^{(d)} \\ 
              & = \IMM_{n,d}
    \end{align*}
where the first two sums are over all cycle covers of $G$ and the third
sum is over all $s,t$-path in the ABP corresponding to $\IMM_{n,d}$
\end{proof}

\bibliographystyle{unsrtnat}
\bibliography{references}  %

\end{document}